\documentclass[12pt]{article}
\usepackage[english]{babel}

\usepackage{amsmath}
\usepackage{times}
\usepackage{graphicx}
\usepackage{color}
\usepackage{multirow}
\usepackage{amssymb}
\usepackage{amsthm} 
\usepackage{mathrsfs} 
\usepackage{sectsty}
\usepackage{epstopdf}
\usepackage{enumerate}
\usepackage{longtable}
\usepackage{bm}
\usepackage{subcaption}

\DeclareMathAlphabet{\pazocal}{OMS}{zplm}{m}{n}
\usepackage{calrsfs}

\usepackage[latin9]{inputenc}

\usepackage[T1]{fontenc}
\usepackage[toc,page]{appendix}
\usepackage{varioref}
\usepackage{float}
\usepackage{titlesec}

\usepackage[authoryear]{natbib}

\usepackage{rotating}
\usepackage{bbm}
\usepackage{latexsym}
\usepackage{mathtools}
\usepackage{afterpage}
\usepackage[bottom]{footmisc}

\DeclareGraphicsExtensions{.eps,.png}

\sectionfont{\normalsize} 
\subsectionfont{\normalsize}

\newcommand{\captionfonts}{\normalsize}

\makeatletter  
\long\def\@makecaption#1#2{%
  \vskip\abovecaptionskip
  \sbox\@tempboxa{{\captionfonts #1: #2}}%
  \ifdim \wd\@tempboxa >\hsize
    {\captionfonts #1: #2\par}
  \else
    \hbox to\hsize{\hfil\box\@tempboxa\hfil}%
  \fi
  \vskip\belowcaptionskip}
\makeatother   

\newtheorem{theo}{Theorem}

\newtheorem{cor}{Corollary}
\newtheorem{remark}{Remark}

\newcommand{\re}{\mathbb{R}}

\newcommand{\cL}{\mathbfcal{L}}    	
\newcommand{\cF}{\mathbfcal{F}}		

\newcommand{\cf}{\pazocal{F}}

\DeclareMathAlphabet\mathbfcal{OMS}{cmsy}{b}{n}

\newcommand{\cR}{\pazocal{R}}

\newcommand{\cP}{\pazocal{P}}

\newcommand{\Sp}{\mathbb{S}}

\newcommand{\bs}{\boldsymbol}

\begin{document}
\hspace{13.9cm}1

\allowdisplaybreaks

\ \vspace{20mm}\\

{Learning with precise spike times: A new decoding algorithm for liquid state machines}

\ \\
{\bf Dorian Florescu, Daniel Coca$ ^* $}\\
Department of Automatic Control and Systems Engineering, The University of Sheffield, Sheffield, S1 3JD, UK.\\
%


\thispagestyle{empty}
\markboth{}{NC instructions}
\ \vspace{-0mm}\\
%

\begin{center} {\bf Abstract} \end{center}



There is extensive evidence that biological neural networks encode information in the precise timing of the spikes generated and transmitted by neurons, which offers several advantages over rate-based codes. 
Here we adopt a vector space formulation of spike train sequences and introduce a new liquid state machine (LSM) network architecture and a new forward orthogonal regression algorithm to learn an input-output signal mapping or to decode the brain activity. The proposed algorithm uses precise spike timing to select the presynaptic neurons relevant to each learning task.
We show that using precise spike timing to train the LSM and selecting the Readout presynaptic neurons leads to a significant increase in performance on binary classification tasks, in decoding neural activity from multielectrode array recordings, as well as in a speech recognition task, compared with what is achieved using the standard architecture and training methods. 

\section{Introduction}

It is generally accepted that neurons in the brain encode information not only in their average firing rates - rate coding - but also in the precise timing of spikes - temporal coding \citep{Hirata2008}. The importance of the precise spike timing information  has been documented in many studies \citep{Srivastava2017,Memmesheimer2014,Kayser2009,Jones2004,Gollisch2008,Riehle1997}. \cite{Seth2015} has argued that the two encoding schemes are in fact complementary.

Neuronal coding is reproducible with a precision of a millisecond \citep{Mainen1995,izhikevich2006}. It has been argued that codes that utilise spike timing make better use of the capacity of neural connections than those relying on rate codes \citep{Mainen1995} and that it allows processing information on much shorter time scales allowing to track rapidly changing signals \citep{Gardner2016}.

There is also evidence that during perceptual decisions, learning and behaviour can be driven by a small number of neurons that are trained to read out and interpret very sparse, precisely timed action potentials \citep{Huber2008,Houweling2008,Wolfe2010}.

In recent years, a lot of research effort has been expanded to establish a sound theoretical basis for encoding and decoding using the precise timing of the spikes rather than spike-count rates \citep{faithful,Florescu2015,lazar2015spiking,Florescu2017,Florescu2018}. A range of supervised learning approaches that utilise temporal coding schemes have been developed for recurrent spiking neural networks (SNNs) with feedforward and feedback connections \citep{Gardner2016,Gutig2014}.
Some of the popular SNN training algorithms using temporal coding are based on gradient descent \citep{Bohte2002,Xu2013,Florian2012,Pfister2006} or on spike timing dependent plasticity \citep{Pfister2006,Florian2007,izhikevich2007,Ponulak2010}.

Liquid state machines (LSM) \citep{Maass2002} are a class of recurrent SNNs that consist of a fixed high-dimensional dynamical network of biologically-realistic synapses and spiking neurons that remain unchanged during training, known as reservoir or 'Liquid', followed by a memoryless output or 'Readout' unit with adjustable synaptic weights. The Readout typically combines in a linear fashion the outputs of all the neurons in the Liquid. The LSM model can be viewed as a nonlinear dynamical system where the state vector comprises the states of all neurons in the Liquid, evolving in time according to the internal dynamics and external driving inputs, and the static Readout defines the relationship between the state vector and output \citep{Maass2002}.

The LSMs belong to the general class of reservoir computing approaches, which, compared with high-dimensional recurrent neural networks, have more biologically plausible architectures and simpler training algorithms that only tune the weights of the connections to the Readout unit \citep{Luko2009}.

The reservoir computing approaches also include non-spiking models, as the Echo State Networks (ESNs) \citep{Jaeger2001}. However, the LSMs are more biologically realistic than ESNs and thus better suited for reproducing the computational properties of biological neural circuits.

The LSM Readout is typically trained by performing linear regression using the spike train outputs of the Liquid converted to continuous signals with exponential filters \citep{Maass2002}. 
Other proposed LSM models have feedback connections from the Readout, and are trained with recursive least squares using the filtered outputs of the Liquid \citep{Clopath2017}. This leads to losing the information of the exact spike times generated by the Liquid neurons. The current training methods for LSMs learn target outputs using measurements from all the presynaptic neurons \citep{Maass2002,isolated2005} . Numerically, this model contains a large number of parameters which can lead to overfitting for large neural circuits. Moreover, it is known that only a relatively small number of cortical neurons project to different areas of the central nervous system  \citep{Hausler2007,Thomson2002}.  

In the case of ESNs, \cite{Dolinsky2017} used orthogonal forward regression (OFR) to identify the contribution of each individual neuron to the response variable, and concluded that a small number of presynaptic neurons are enough to achieve accurate results. 

Here we propose a new liquid state machine (LSM) architecture, and a new training
algorithm that outperforms the standard methods \citep{Maass2002,isolated2005}. The architecture consists of a Liquid, comprising only a SNN, in series with a spike time based Readout. The new algorithm, called  OFR with Spike Trains (OFRST), identifies the best synaptic connectivity for the Readout unit of the LSM. The learning algorithm relies on a distance metric between two spike trains that are elements of an inner product vector space \citep{Carnell2005}.

Theoretical results demonstrate that the proposed architecture can learn any continuous target output by mapping it onto a unique target spike train sequence. We prove that the proposed LSM architecture achieves higher accuracy in training compared with the standard methods.

Numerical simulations are given to show the performance of the proposed method
compared to the standard methods for binary and multi-label input classification tasks.
Additional numerical examples are used to show separately the benefit of selecting the Readout connectivity using OFR and computing with precise spike times.
The advantage of the proposed method is also demonstrated for two problems involving real world data. First we consider the problem of classifying the movement direction of drifting sinusoidal gratings using visually evoked multi-array recordings from the primary visual cortex of the monkey. 
Second, we test our method against the standard methods on a problem of speech recognition.

The paper is structured as follows. Section 2 introduces the standard architecture
and method for training an LSM. Section 3 presents the proposed approach. Numerical simulations are in Section 4. Section 5 presents the conclusion.

\section{The Standard LSM Architecture and Training Method}
\label{sec:standard}

The spike train inputs and outputs of the LSM are elements of space $ \Sp_0 $ satisfying 
\begin{equation*}
	\Sp_0=\left\lbrace s \vert s=\{t_k\}_{k=1}^P, t_{k+1}>t_k\geq0,\forall k=1,\dots,P-1  \right\rbrace
\end{equation*}

The Liquid is modelled by an operator $ \cL $, which maps the vector of input spike trains $ \bs{s}^{in} $ into a vector of continuous functions $ \bs{x}(t) $, also known as the state of the Liquid. The Readout is modelled by operator $ \cR $, which maps $ \bs{x}(t) $ into the continuous scalar function $ y(t) $, which denotes the LSM output. The function $ y(t) $ satisfies \citep{Maass2002}
\begin{equation*}
y(t)=\cR\left(\cL \bs{s}^{in}\right),
\end{equation*}
where $ \bs{s}^{in}=\left[s_1^{in},\dots,s_{N_{in}}^{in}\right] $ , $ s^{in}_k=\left\lbrace t_{k,1}^{in},\dots,t_{k,P_k^{in}}^{in} \right\rbrace $, $ \cL:\left[\Sp_0\right]^{N_{in}}\rightarrow \left[L^2(\re)\right]^N, \cR:\left[L^2(\re)\right]^{N}\rightarrow L^2(\re) $, where $ N_{in} $ and $ N $ denote the number of inputs and number of neurons in the SNN, respectively, $ P_k^{in} $ denotes the number of spikes in input $ k $, and $ \bs{x}(t) = \left(\cL \bs{s}^{in}\right)(t) $.

The Liquid is represented as the composition of two mathematical operators $ \cL=\cF \cL_{SNN} $, where $ \cL_{SNN}:\left[\Sp_0\right]^{N_{in}}\rightarrow\left[\Sp_0\right]^N $ models a generic SNN and $ \cF:\left[\Sp_0\right]^N\rightarrow [L^2(\re)]^{N} $, $ \cF \bs{s} = [\cf s_1, \cf s_2, \dots, \cf s_N] $, $\forall \bs{s}\in\left[\Sp_0\right]^N, \bs{s}=[s_1,\dots,s_N] $ models a pool of linear filters
\begin{equation}
\label{eq:filters}
\cf s_n = \sum_{k=1}^{P_n}e^{-\frac{t-t_k^n}{\tau_s}}\cdot 1_{[t_k^n,\infty)}(t),
\end{equation}
where $ P_n $ denotes the number of spikes in $ s_n $, $ 1_{[t_k^n,\infty)} $ denotes the characteristic function of interval $ [t_k^n,\infty) $, and $ \tau_s $ denotes the time constant of the filter.

\cite{Maass2002} demonstrated that this model has, under idealised conditions, universal real-time computing power. The standard LSM architecture is presented in Figure \ref{fig:LSM_standard_diagram}.
\begin{figure}[!ht]
	\hfill
	\begin{center}
		\includegraphics[width=4.5in,trim={0cm 0cm 0cm 0cm},clip]{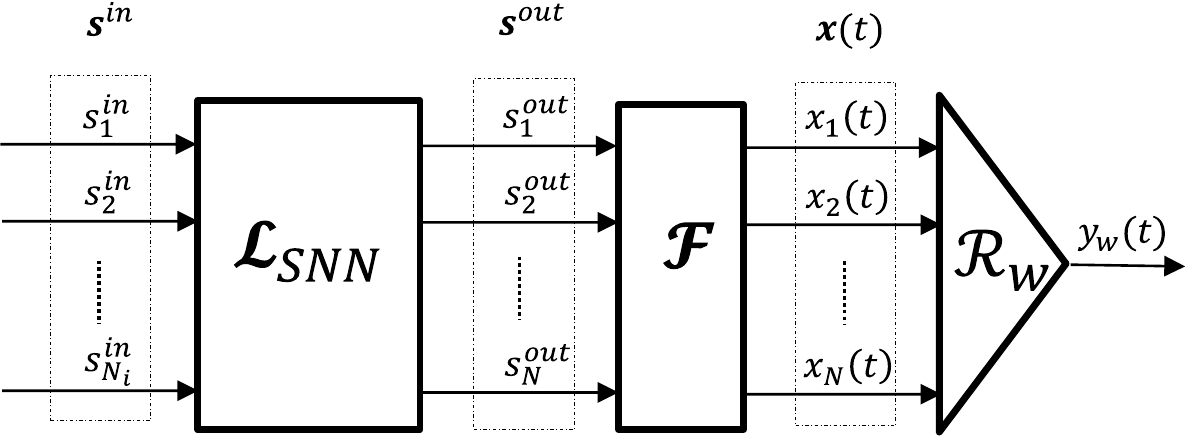}
	\end{center}
	\caption{Block diagram of the standard architecture used for training LSMs. It consists of three blocks connected in series: the Liquid $ \cL_{SNN} $, the pool of filters $ \cF $ and the readout $ \cR_w $.}
	\label{fig:LSM_standard_diagram}
\end{figure}

\begin{remark} \label{rem:lin_ind}
	Throughout the paper, it will be assumed that $ s^{out}_k \neq s^{out}_l, $  $\forall k,l\in\{1,\dots,N\},$  $k\neq l $. In a practical scenario it is very unlikely that two neurons will generate two identical spike trains simultaneously. However, if this happens to be true, only the distinct outputs will be used for training. 
\end{remark}

The most common Readout is the linear unit $ \cR_{\bs{W}} \bs{x}(t)= \sum_{n=1}^{N} w_n x_n(t)$, where $ \bs{W}=[w_1,\dots,w_N] $ and $ \bs{x}(t)=[x_1(t),\dots,x_N(t)] $. This Readout was shown to classify time-varying inputs with the same power as complex non-linear Readouts, given a large enough Liquid \citep{Hausler2002}. A typical way to train the Readout is by tuning the weights using the least squares (LS) algorithm
\begin{equation}
 \bs{w}_{opt}= \underset{\bs{w}}{\text{argmin}} \| y^*-  y_{\bs w} \|_{L^2}, 
\end{equation} 
where $ y^* \in L^2(\mathbb{R}) $ denotes the target output function, $ \|\cdot\|_{L^2} $ denotes the standard norm in $ L^2(\mathbb{R}) $ and $ y_{\bs w}=\cR_{\bs{w}} \cF\cL_{SNN} \bs{s}^{in} $ denotes the predicted output.

In practice, the continuous state of the liquid $ \bs x(t) $ is sampled uniformly with period $ \Delta T>0 $. The function $ \bs x(t)=\left[ \cf s_1^{out}, \cf s_2^{out}, \dots,\cf s_N^{out}\right] $ is not continuous in a mathematical sense at points $ \{t_k^n\}_{k=1}^{P_n},n=1,\dots,N $,  due to the expression of operator $ \cf $ \eqref{eq:filters}. Therefore, for any sequence of spike trains $ \{s_1^{out},\dots,s_N^{out}\} $,  $ \bs x(t) $ in not bandlimited. This can also be explained by viewing the values of operator $ \cf $ as the output of an exponential filter with impulse response $ h(t)=e^{-\frac{t}{\tau_s}} $, given a train of Dirac delta pulses $ \sum_{k=1}^{P_n} \delta(t-t_k^n)$. Given that the filter is not ideal, its output has arbitrarily large frequency components, and thus the samples $ \left\lbrace\bs{x}(kT)\right\rbrace $ are affected by aliasing, due to Shannon's law. This leads to computing weights $ \bs{w}_{opt} $ that are deviated form the theoretical optimal values, as well as an imprecise final output prediction $ y_{\bs{w}_{opt}}(t) $.

Moreover, in practice not all synaptic connections of the Readout are relevant to a particular task, so that training the weights of all possible connections from the Liquid neurons to the Readout can easily lead to overfitting. 

There are a few variations of LS that introduce an additional parameter, also known as hyperparameter, in order to control the effective complexity of the model and to reduce overfitting. Some of the standard methods doing this are LS with $ L^2 $ regularization, or ridge regression (RR), LS with $ L^1 $ regularization, or lasso, and early stopping (ES). The regularization parameter for RR and lasso, and the number of iterations for ES are typically tuned to minimise the prediction error on the validation dataset \citep{Bishop}. These methods can lead to a Readout with smaller weights, or fewer presynaptic connections to the Liquid.

However, computing the Readout weights with RR, lasso or ES is affected by approximation error, as a result of the aliasing effect caused by uniform sampling. This leads to Readout presynaptic connections to neurons that are less relevant for the computing task. Furthermore, the output spikes of a biological neural network do not lie on a grid of uniformly spaced time points, and therefore are not directly compatible with the training methods above.



\section{A New LSM Training Approach using Precise Times}

\subsection{The Carnell-Richardson Spike Train Space}

The space $ \Sp_0 $ is not a linear space because it does not allow any operations between spike trains. To overcome this problem, this space is extended to the Carnell-Richardson spike train space \citep{Carnell2005}
\begin{equation*}
\Sp=\left\lbrace s=\left\lbrace \left(a_k,t_k\right) \right\rbrace_{k=1}^P, P\geq 1, t_k, a_k \in \re, t_k\neq t_l, \forall k,l\in\{1,\dots,P\}, k\neq l \right\rbrace.
\end{equation*}

\cite{Carnell2005} have proven that $ \Sp $ is an inner product space, where the vector sum, scalar multiplication and inner product of two spike trains $ s_1,s_2\in\Sp $ are defined as
\begin{align*}
s_1+s_2&=\{(a_k^1,t_k^1)\}_{k=1}^{M_1}\cup \{(a_k^2,t_k^2)\}_{k=1}^{M_2},\\
\alpha \cdot s&=\{(\alpha\cdot a_k,t_k)\}_{k=1}^M, \forall \alpha \in \re,\\
\left\langle s_1,s_2\right\rangle_{\Sp} &= \sum_{k_1=1,k_2=1}^{k_1=M_1,k_2=M_2}a_{k_1}^1a_{k_2}^2\cdot e^{-\frac{|t_{k_1}^1-t_{k_2}^2|}{\tau_s}},
\end{align*} 
where $ \tau_s>0 $ is a scaling factor. The inner product $ \langle \cdot,\cdot \rangle_{\Sp} $ generates a norm $ \|\cdot\|_{\Sp} $ satisfying $ \|s\|_{\Sp}=\sqrt{\langle s,s \rangle_{\Sp}}, \forall s \in \Sp $. Figure \ref{fig:Carnell} illustrates an example of a linear operation between two randomly generated spike trains $s_1,s_2\in \Sp $, presented comparatively with the equivalent operation in $ L^2(\mathbb{R}) $.

\begin{figure}[!ht]
	\hfill
	\begin{center}
		\includegraphics[width=6in,trim={1cm 0cm 0cm 0cm},clip]{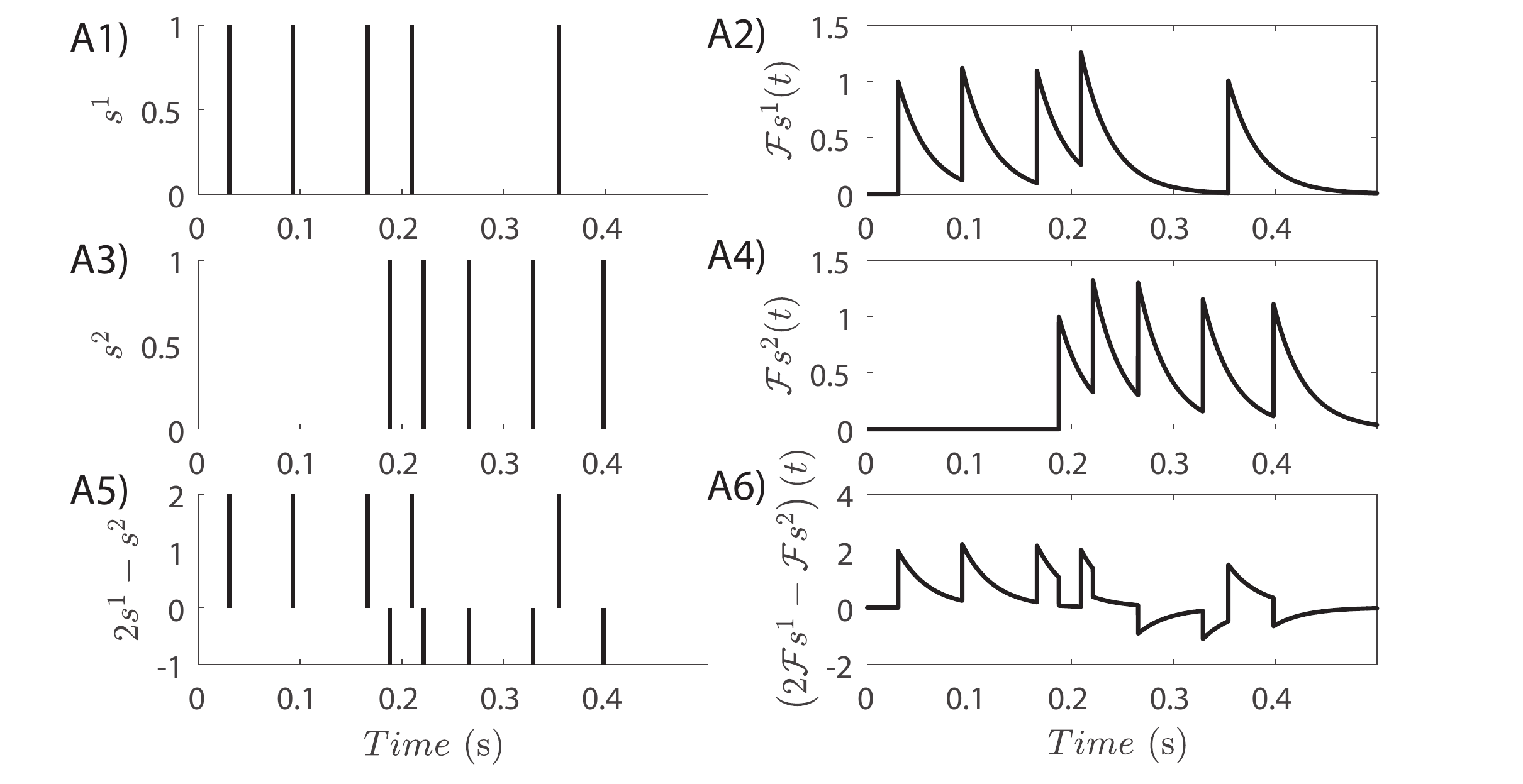}
	\end{center}
	\caption{An example of a linear operation in $ \Sp $. Two spike trains $ s^1 $, $ s^2 \in \Sp $, and their corresponding elements $ \cf s^1, \cf s^2 \in L^2(\mathbb{R}) $, are generated in time interval $ [0,0.5\ \text{s}] $ (A1-4). The equivalent linear operations in the two spaces $ 2s^1-s^2 $ and $ 2 \cf s^1-\cf s^2 $ are depicted in (A5-6).}
	\label{fig:Carnell}
\end{figure}

A spike train $ s_0=\{t_k\}_{k=1}^P\in \Sp_0 $, as defined by the standard method, can be mapped uniquely onto an element $ s\in\Sp $, such that $ s=\left\lbrace \left(1,t_k\right) \right\rbrace_{k=1}^P $. \cite{Maass2002} have defined a metric $ d $ on $ \Sp_0 $
\begin{equation*}
d(s_1,s_2)=\left[\int_{\re}\left[\left(\cf s_1\right)(t)-\left(\cf s_2\right)(t) \right]^2 dt\right]^{1/2},
\end{equation*}
where $ \cf:\Sp_0\rightarrow L^2(\re), \cf s= \sum_{k=1}^{P}e^{-\frac{t-t_k}{\tau_s}}\cdot 1_{[t_k^n,\infty)}(t)$ denotes the output of a linear filter with exponential decay and time constant $ \tau_s, $ given spiking input $ s. $  The norm $ \|\cdot\|_{\mathbb{S}} $ relates to metric $ d $ as follows $ \|s_1-s_2\|_{\mathbb{S}}^2=2\cdot d(s_1,s_2)^2, \forall s_1,s_2 \in \Sp_0. $ However, in a practical setting, the metric $ d $ is approximated by $ d_{\Delta T} $, computed on a uniform grid with sampling time $ \Delta T. $ Then the following holds

\begin{equation*}
\lim_{\Delta T \rightarrow 0} d_{\Delta T} (s_1,s_2)=\frac{1}{\sqrt{2}} \| s_1-s_2\|_{\Sp}.
\end{equation*}

In order to show the disadvantage in computing $ d_{\Delta T} $, we generated two random spike trains $ s_1 $ and $ s_2 $ with $ 100 $ spike times each. We then computed $ \| s_1-s_2 \|_{\Sp} $ and $ d_{\Delta T}(s_1,s_2) $ for $ 100 $ values of $ \Delta T $ on $ [1\ \text{ms}, 100\ \text{ms}] $, and $ \tau_s=30\ \text{ms} $. The results, depicted in Figure \ref{fig:norm_comparison}, show that the values of $ d_{\Delta T}(s_1,s_2)/\sqrt{2} $ oscillate around $ \| s_1-s_2 \|_{\Sp} $ as $ \Delta T \rightarrow 0. $ However, the computing time for $ d_{\Delta T} $ increases exponentially with $ 1/\Delta T. $ Thus, at the sampling interval of $ 2\ \text{ms} $, which is used to simulate the LSM, the spike based metric results in a similar value to the standard metric, but runs three times faster.

\begin{figure}[!ht]
	\hfill
	\begin{center}
		\includegraphics[width=5in,trim={0cm 0cm 1cm 1cm},clip]{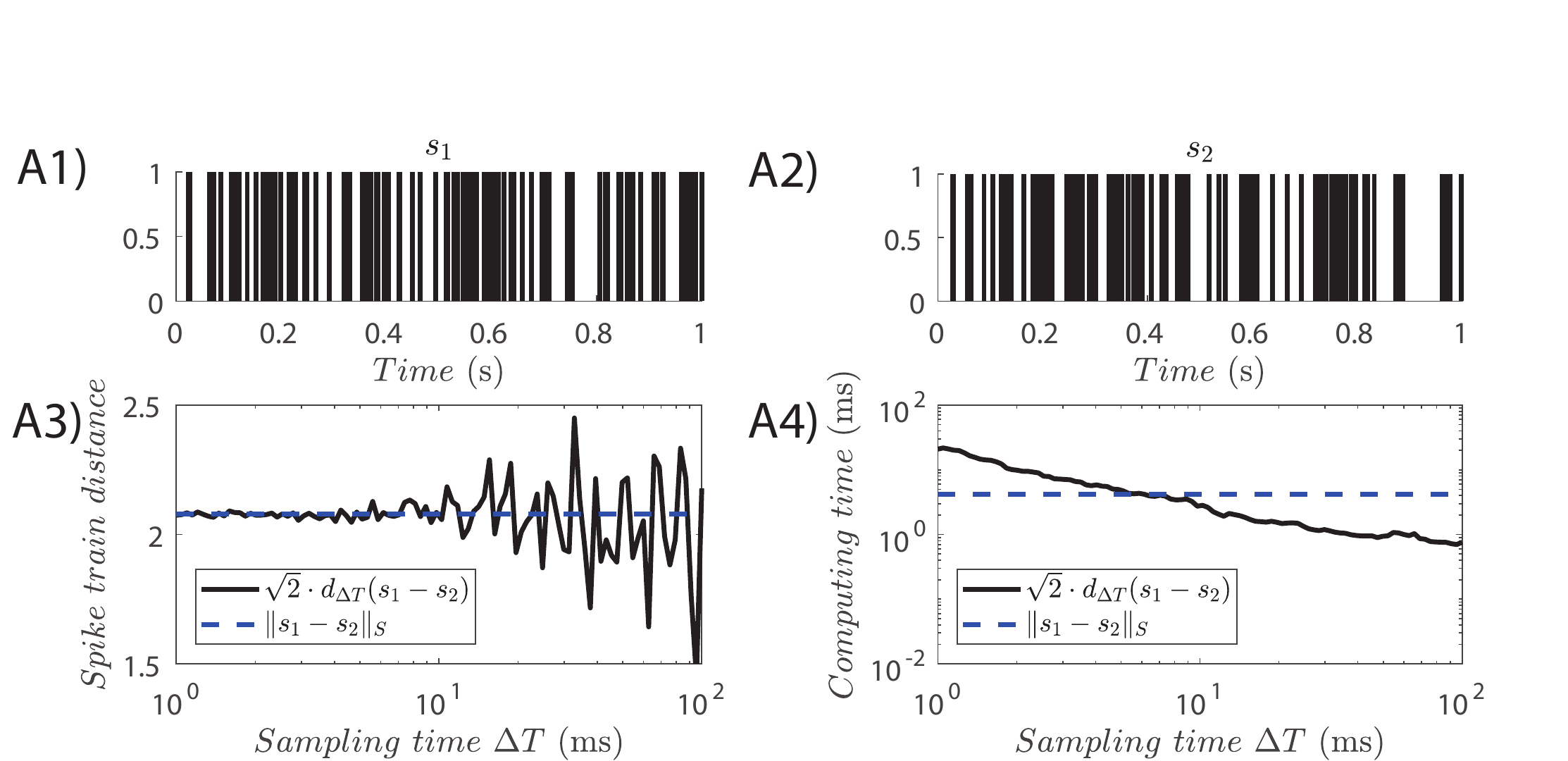}
	\end{center}
	\caption{Comparison between the Carnell-Richardson spike train distance $ \|s_1-s_2\|_{\Sp} $ and the standard metric $ d_{\Delta T}(s_1-s_2) $: two randomly generated spike trains $ s_1, s_2 $ (A1,2) and their corresponding distance calculated with the two metrics (A3,4).}
	\label{fig:norm_comparison}
\end{figure}

\subsection{The Proposed LSM Architecture and Training Method}

We propose a new spike time based Readout architecture, which does not require the bank of filters $ \cF $ (Figure \ref{fig:LSM_proposed_diagram}).
\begin{figure}[!ht]
	\hfill
	\begin{center}
		\includegraphics[width=3.5in,trim={0cm 0cm 0cm 0cm},clip]{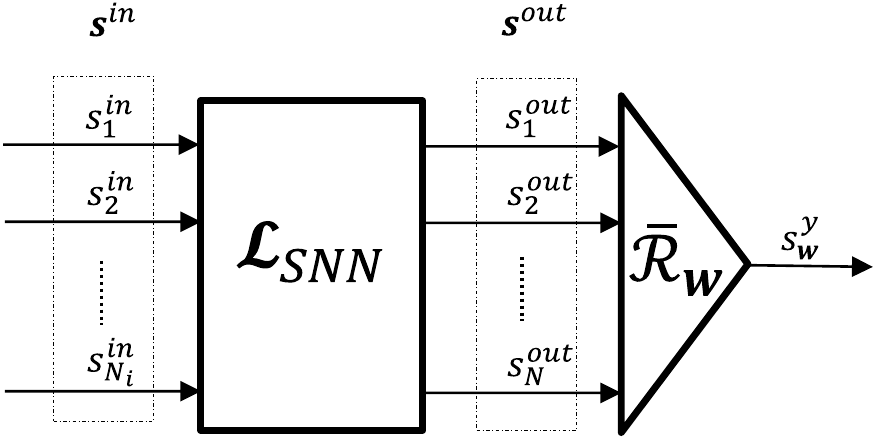}
	\end{center}
	\caption{Block diagram of the proposed architecture used for training LSMs , consisting of two blocks connected in series: the Liquid $ \cL_{SNN} $ and the proposed spike based Readout $ \bar{\cR}_w $.}
	\label{fig:LSM_proposed_diagram}
\end{figure}

The Readout $ \bar{\cR}_{\bs w} $ is defined using the operations in $ \Sp $ as
\begin{equation*}
\bar{\cR}_{\bs w}\bs{s}^{out}=\sum_{n=1}^{P_n^{out}}w_n s_n^{out}=s_{\bs w}^y.
\end{equation*}

 Let $ s^{y*} $ be a target spike train. Then the optimal $ \bs w $ in the least squares sense is
\begin{equation*}
\bar{\bs w}_{opt}=\underset{\bs w}{\text{argmin}}\| s^{y*}- s_{\bs w}^{y} \|_{\Sp},
\end{equation*}
where $ \| \cdot \|_{\Sp} $ denotes the standard norm in $ \Sp $.
 
The proposed architecture can be extended to learn continuous target signals. To this end, the following results demonstrate that any continuous target function $ y^*\in L^2(\re) $ can be mapped uniquely onto a spike train $ s^{y*}\in \Sp $.
 
\begin{theo} \label{theo:mapping}
Let $ \Sp^{out} $ denote the subset of $ \Sp $ generated by the outputs of the SNN, such that $  \Sp^{out}=\textup{span}\{s_1^{out},\dots,s_N^{out}\}\subset\Sp. $ 
Let $ \cf:\Sp^{out}\rightarrow L^2(\re)$ be an operator defined by 
\begin{equation}
\label{eq:filters2}
\cf s = \sum_{k=1}^{P}a_k e^{-\frac{t-t_k}{\tau_s}}\cdot 1_{[t_k,\infty)}(t), \forall s \in \Sp^{out}, s= \{(a_k,t_k)\}_{k=1}^P.
\end{equation}

Moreover, let $ \mathbb{F}\mathbb{S}^{out} $ denote the subset of $ L^2(\re) $ generated by the filtered outputs of the SNN, such that $ \mathbb{F}\mathbb{S}^{out}=\textup{span} \{\cf s_1^{out},\dots,\cf s_N^{out}\}. $ Then the following mapping is well defined
\begin{equation} \label{eq:mapping}
\pazocal{M}:L^2(\re)\rightarrow\Sp^{out}, \pazocal{M}(y)=\cf^{-1}\pazocal{P}_{\mathbb{F}\mathbb{S}^{out}} (y), \forall y \in L^2(\re),
\end{equation}
where $ \cP $ denotes the projection operator.
\end{theo}
\begin{proof}
	See Appendix 1.
	\end{proof}

Theorem \ref{theo:mapping} defines a mapping that allows converting any continuous target output function $ y^*(t) $ into a unique target output spike train $ s^{y*}. $
The operator $ \cf $ in \eqref{eq:filters2} is the extension of the filtering operator in \eqref{eq:filters} to the more general space $ \Sp $.
The following result assesses the prediction accuracy of the proposed method relative to the standard method for continuous target functions.

 \begin{theo} \label{theo:optimum}
 	Let $ y^*\in L^2(\re) $ and let $ \bs w_{opt} $ be the vector of weights computed for the standard architecture, such that $ \bs{w}_{opt} =\underset{\bs{w}}{\textup{argmin}} \| {y^*} - {\cR}_{\bs{w}} \cF  \bs s^{out}\|_{L^2} $. It follows that
 	\begin{equation*}
 	\bs{w}_{opt} =\underset{\bs{w}}{\textup{argmin}} \| s^{y*} - \bar{\cR}_{\bs{w}}  \bs s^{out}\|_{\Sp}=\bar{\bs{w}}_{opt}, 
 	\end{equation*}
where $ s^{y*}=\pazocal{M}(y^*) $, $ \pazocal{M}(y^*)=\cf^{-1}\pazocal{P}_{\mathbb{F}\mathbb{S}^{out}} (y^*) $, $ \cP $ denotes the projection operator and $ \mathbb{F}\Sp^{out} = \textup{span} \left\lbrace \cf s_1^{out},\dots,\cf s_N^{out} \right\rbrace. $
 \end{theo}
\begin{proof}
See Appendix 1.
\end{proof}
 
\begin{cor}
	\label{cor:1}
	Theorem \ref{theo:optimum} proves that the proposed methodology achieves, in theory, the same accuracy as the state-of-the-art method when learning continuous target signals. In practice, however, the accuracy of the standard method is lower because it is affected by the approximation error introduced when calculating $ \bs{w}_{opt} $ and $ y_{\bs{w}_{opt}}(t) $ from uniform samples, which doesn't affect the proposed method.
\end{cor}

\subsection{The Orthogonal Forward Regression with Spike Trains (OFRST) Algorithm}
\label{subs:OFRST}

The optimisation problem addressed by the proposed method is to learn a continuous target output $ y^*(t) $ given a SNN of size $ N $. Let $ \{s_{k}^{out}\}_{k=1}^{N} $ denote the outputs of the SNN in response to stimuli $ \{s_{k}^{in}\}_{k=1}^{N_{in}} $. Computing the optimal  $ \bs{w}_{opt} $ in the least squares sense \citep{Maass2002} leads to many non zero weights that are not particularly relevant for the learning task and overfit the data. Furthermore, the standard methods that address this problem using regularization or early stopping lead to weights that are deviated from the theoretical optimal weights as a result of the approximation error.

Theorem \ref{theo:mapping} demonstrates that the problem addressed here can be reduced to learning a target spike train $ s^{y*} $, uniquely derived from the continuous target $ y^*(t). $ This leads to a more precise estimation of weights $ \bs{w}_{opt} $ (Theorem \ref{theo:optimum}). Here we introduce a greedy selection algorithm for the spike trains that are most relevant for the learning task, called Orthogonal Forward Regression with Spike Trains (OFRST). The  OFRST algorithm is inspired by the orthogonal forward regression (OFR) for finite dimensional spaces \citep{chen1989OFR}. The remaining part of this section will first present the classical OFR and then the proposed OFRST algorithm.

Given vectors $ \{x_1,\dots,x_N \} $ and target vector $ y^* $, the OFR algorithm aims to identify a subset $ \{x_{\ell_1}\dots,x_{\ell_p} \} $ and an estimate of the parameters $ \{w_{\ell_1},\dots,w_{\ell_p} \} $ that fits the data $ y^* $. 

At the first stage, $ y^* $ is projected onto basis vectors $ \{x_1,\dots,x_N \} $. Then the error-reduction-ratio (ERR) is calculated for each vector, defined as
\begin{equation*}
ERR_k^{(1)}=\frac{\langle x_k,y^* \rangle^2}{\| x_k \|^2 \cdot \| y^* \|^2}.
\end{equation*}
The magnitude of $ ERR_k^{(1)} $ represents the proportion of the dependant variable variance explained by $ x_k. $ A geometrical interpretation of the ERR is depicted in Figure \ref{fig:OFR_intuition} for the simplified case where $ x_k\in\mathbb{R}^2, k=1,2, $ and $ y^* \in \mathbb{R}^2. $
The maximum ERR, computed as $ ERR_1=ERR_{\ell_1}^{(1)}=\max_{k=1,\dots,N}\{ ERR_k^{(1)} \} $, leads to the selection of $ x_1^{\perp}=x_{\ell_1} $ as the basis for the one-dimensional space $ E^1. $ 

\begin{figure}[!ht]
	\hfill
	\begin{center}
		\includegraphics[width=3.5in,trim={0cm 0cm 0cm 0cm},clip]{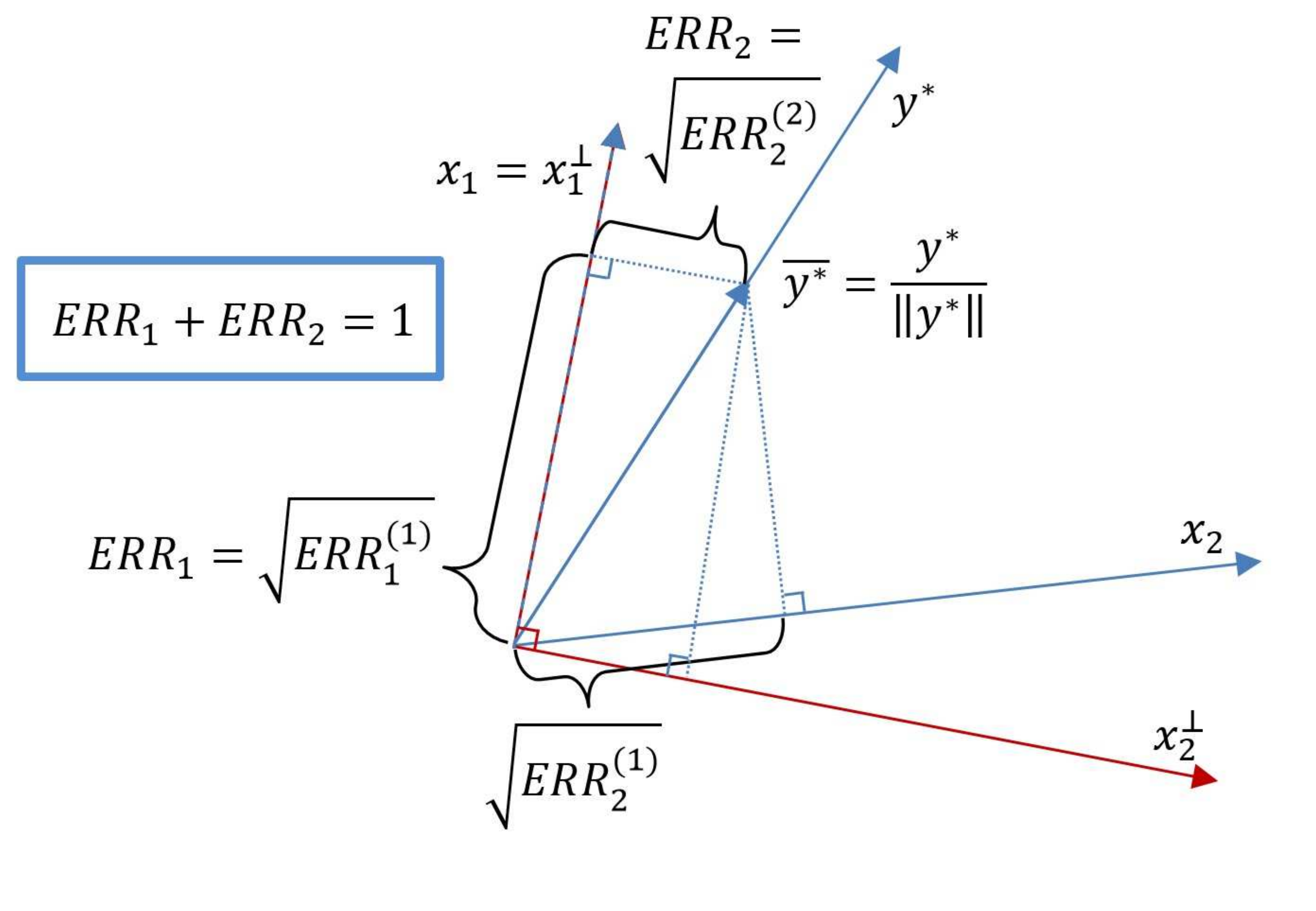}
	\end{center}
	\caption{Geometrical interpretation of OFR for the simplified two-dimensional scenario. In this case $ ERR_1^{(1)}>ERR_2^{(1)} $ implies that $ x_1 $ explains a larger proportion of the variance of target output $ y^*. $ }
	\label{fig:OFR_intuition}
\end{figure}

At the second stage, the rest of the vectors $ \{x_i\}_{i=1,\dots,N,i\neq \ell_{1}} $ are projected, through Gram-Schmidt orthogonalization, into a $ (N-1) $-dimensional space orthogonal on $ E_1 $. Subsequently, the vector $ x_{\ell_2} $ is selected and orthogonalised with the Gram-Schmidt procedure to compute $ x_2^{\perp} $. The vectors $ x_1^{\perp} $ and $ x_2^{\perp} $ form the basis for two-dimensional space $ E_2. $ Similarly, at stage number $ p $, the vector $ x_{\ell_p} $ is selected, which is used to define the $ p $-dimensional space $ E_p $ with orthogonal basis $ \{x_{i}^{\perp}\}_{i=1\dots,p} $. The detailed algorithm is given in Appendix 2.

The OFRST algorithm closely follows the steps of the OFR algorithm, implemented for the Carnell-Richardson spike train space $ \Sp. $ Initially, let $ s_{k}^{\perp(1)}=s_{k}^{out}\in \Sp,\forall k=1,\dots,N, $ be the complete set of SNN outputs. The most significant spike train $ s_{\ell_1}^{out} $ is defined as the one that maximises $ ERR_k^{(1)} $, where $ ERR_k^{(i)} $ denotes the error-reduction-ratio (ERR) of term $ k $ at iteration $ i $,  defined as 
\begin{equation*}
ERR_k^{(i)}=\frac{\left\langle s_k^{\perp(i)},s^{y*} \right\rangle_{\Sp}^2}{\|s_k^{\perp(i)}\|_{\Sp}^2\cdot\|s^{y*}\|_{\Sp}^2}.
\end{equation*}	

Subsequently, the set $ \{s_k^{\perp(2)}\}_{k=1,k\neq\ell_1}^{N} $ is computed by orthogonalising the remaining output spike trains against $ s_{\ell_1}^{out} $ using the Gram-Schmitt routine.

The process continues iteratively. At every iteration $ i $, the algorithm selects the next most significant spike train $ s_{\ell_i}^{out} $ such that $\ell_{i}=\underset{k}{\text{argmax}} \left(ERR_k^{(i)}\right) $, and generates the set $ \{ s_{\ell_1}^{out},\cdots,s_{\ell_i}^{out} \} $ of significant SNN outputs and the corresponding vector of weights $ \bs{w}^{(p)} $. Subsequently, the set $ \{s_k^{\perp(i)}\}_{k=1,k\neq\ell_1,\dots,\ell_i}^{N} $ is computed from the remaining spike trains through orthogonalisation. The process continues until $ p=N $. The final number of presynaptic neurons is selected as the smallest $ p $ that leads to the maximum prediction accuracy on the validation dataset. The detailed algorithm is given in Appendix 3.

\section{Numerical examples}
\label{sec:examples}

The proposed new Readout and associated training algorithm is evaluated in comparison with the standard architecture trained with LS, RR, lasso and ES. 

Additional numerical examples show the advantage of using a spike based Readout and the advantage of selecting the Readout presynaptic neurons using OFRST. The benefit of the proposed method is also demonstrated for two additional examples with real world data. First, OFRST is compared against the standard methods for a multi-label classification problem using multi-array recordings from the primary visual cortex of the monkey. Second, the advantage of the proposed method on a speech recognition task is shown using data from the TI-46 corpus database of spoken digits.

The LSM was simulated using the toolbox described in \citep{toolbox2003}. The Liquid consists of leaky integrate-and-fire neurons, $ 20\% $ of which were randomly selected to be inhibitory \citep{Maass2002}.  The connection probability between neurons $ a $ and $ b $ is defined as $ C\cdot e^{-\left(D(a,b)/L\right)^2} $, where $ D(a,b) $ denotes the Euclidian distance between the neurons, $ L=2 $ is a parameter that controls the average number of connections and the average distance between neurons, and $ C $, depending on whether the neurons are excitatory (E) or inhibitory (I), is $ 0.3 $ (EE) $, 0.2 $ (EI) $, 0.4 $ (IE) $, 0.1 $ (II). The synaptic transmission is given by the dynamic model proposed in \citep{Markram1998}. The input is injected into $ 30\% $ randomly chosen neurons in the Liquid with an input gain of $ 0.1 $. For the standard Readout architecture, the time constant of the exponential filters is $ \tau_s=30 $ms. The LSM was simulated using the default sampling time of $ 0.2 $ms \citep{Maass2002}. The simulations were carried out in
Matlab Version $ 8.6\ (R2015b) $ on a $ 3 $ GHz Intel Core i7-$ 5960 $X 8 core PC workstation.

\vspace{1cm}
\textbf{Example 1. Binary classification - comparison with the standard methods.} 

This example compares the performance achieved by a standard LSM with the Readout parameters estimated using the LS, RR, lasso and ES with that of a LSM comprising a spike-based Readout trained using the proposed OFRST method. The LSM consists of $ 240 $ neurons spatially organised as a lattice with dimensions $ 15 $x$ 4 $x$ 4 $.

The task is to discriminate between two spike train templates using the SNN responses. The templates are two instances of a Poisson point process with rate $ 20$ Hz, depicted in Figure \ref{fig:example1_templates}.

\begin{figure}[!ht]
	\hfill
	\begin{center}
		\includegraphics[width=6in,trim={1.2cm 0.3cm 1cm 0cm},clip]{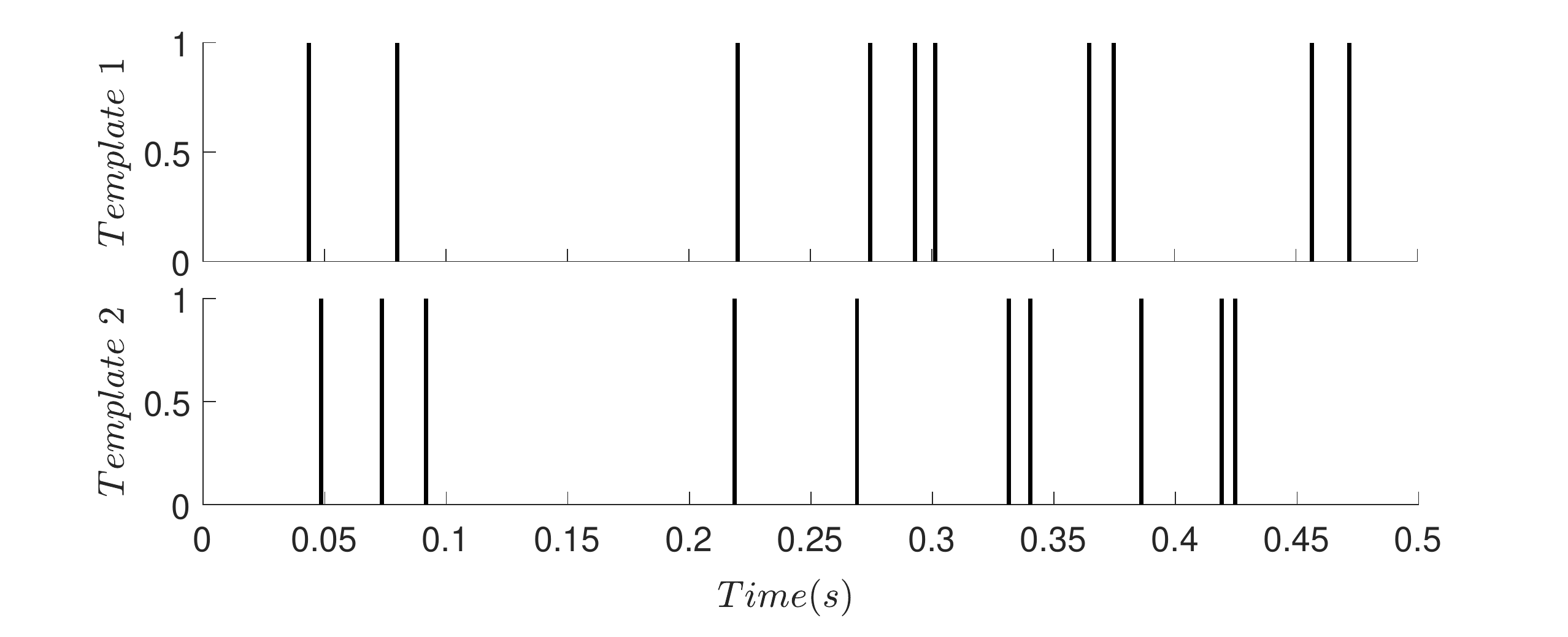}
	\end{center}
	\caption{The input templates used for classification, generated as Poisson spike trains with frequency $ 20 $ Hz over time interval $ [0, 0.5\ \text{s}] $.}
	\label{fig:example1_templates}
\end{figure}

The inputs are generated in time interval $ [0,0.5\ s] $ by jittering one of the two templates, where the jitter noise is drawn from the Gaussian distribution with zero mean and standard deviation $ 6 $ ms. A number of $ 100 $ jittered templates were generated for each class, of which $ 50 $ were used for training and $ 50 $ for validation. The two classes of inputs are assigned the target output labels $ y(t)=1 $ (template $ 1 $) and $ y(t)=-1 $ (template $ 2 $), $ t \in [0,0.5\ s] $. 

The input-output mappings are learned with the LSM by estimating the standard Readout parameters using LS, RR, lasso and ES, where the sampling time is $ \Delta T=20 $ ms \citep{Maass2002,isolated2005}. Subsequently, the spike time based Readout is trained using OFRST. The regularization parameter for RR and lasso, the number of steps for ES and the number $ p $ of presynaptic neurons for OFRST are computed using a line search that maximises the prediction accuracy on the validation dataset.

The classification accuracies for RR, lasso, ES and OFRST were evaluated as a function of the hyperparameter and averaged over $ 100 $ trials. Each trial consisted in a different Liquid and a different instance of jitter applied to the input. The results are depicted in Figure \ref{fig:early_stopping}.

\begin{figure}
	\hfill
	\begin{center}
		\includegraphics[width=6.2in,trim={0.7cm 0cm 0cm 0cm},clip]{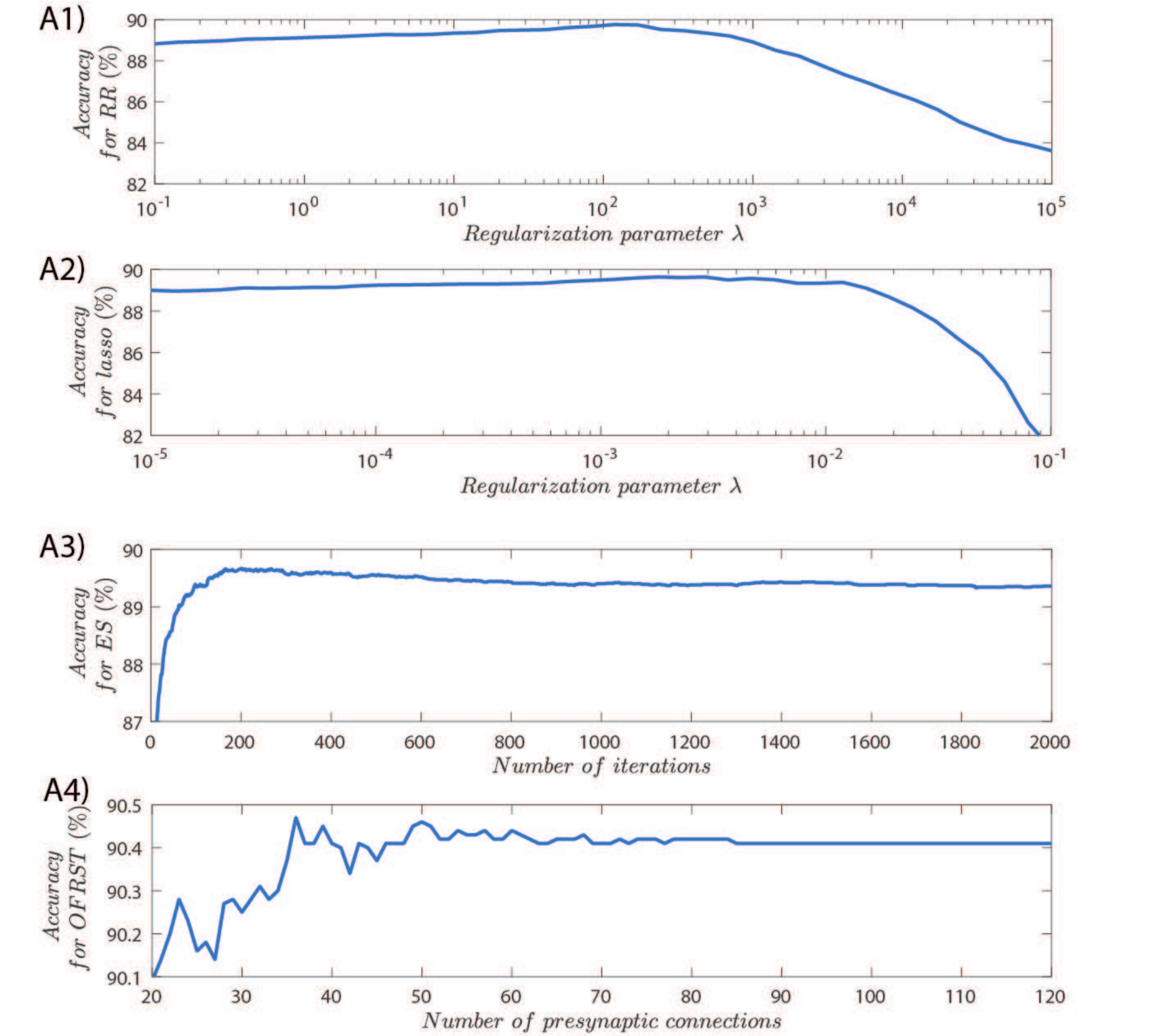}
	\end{center}	
	\caption{Binary classification with RR (A1), Lasso (A2), ES (A3), and OFRST (A4), as a function of the regularization parameter. The average accuracies were computed for each method on the validation dataset over $ 100 $ trials.}
	\label{fig:early_stopping}
\end{figure}

In the case of the OFRST algorithm the results show that, on average, the accuracy drops when using more than $ 36 $ Readout presynaptic connections, or equivalently training for more than $ 36 $ iterations. This suggests that, on average, more than $ 30 $ Readout presynaptic connections lead to overfitting the data. This result mimics what has been observed experimentally in cortical circuits, where only a small number of cortical neurons project to different areas of the central nervous system \citep{Thomson2002,Hausler2007}. 

The accuracy for each method was optimised with a different hyperparameter on each simulation trial. The classification accuracies achieved by all the methods over $ 100 $ trials are given in Table \ref{tb:example1}. The results show that, on average, OFRST has the highest accuracy from all methods while using the smallest number of synapses.

\begin{table}[ht]
	\renewcommand{\arraystretch}{0.55}
	\caption{Binary classification results using pools of $ 240 $ neurons. Comparison between least squares, ridge regression, lasso and early stopping, implemented for the standard Readout, and the proposed OFRST method for the spike based Readout. The mean ($ \pm $ standard deviation) is computed for each method over $ 100 $ trials.}
	\begin{center}
		\begin{tabular}{|c|c||c|}
			\hline
			Training method & Total number of Readout connections  &  Accuracy\\
			\hline
			Least squares & $ 56.46\ (\pm20.3) $ & $ 88.4\%\ (\pm7.28\%) $  \\
			Ridge regression & $ 56.46\ (\pm20.3) $ & $ 91.27\%\ (\pm7.07\%) $\\
			Lasso & $ 40.32\ (\pm23.07) $ & $ 91.15\%\ (\pm6.8\%) $\\
			Early stopping & $ 56.46\ (\pm20.31) $ & $ 91.28\%\ (\pm7.07\%) $\\
			OFRST & $ 15.05\ (\pm11.03) $ & $ 92.15\%(\pm6.92\%) $\\			
			\hline
		\end{tabular}
	\end{center}
	\label{tb:example1}
\end{table}

\vspace{1cm}
\textbf{Example 2. Binary classification - benefits of learning with exact spike times.} 

In this example we compare the classification accuracy of the proposed Readout trained with the OFRST method to that of the standard Readout trained with LS, RR, lasso, ES and classical OFR \citep{Billings1989} on the same binary classification task as in Example 1, but for different values of the sampling time $ \Delta T. $

The training and validation datasets were generated as in Example 1. For the OFR and OFRST methods the number the presynaptic neurons, which represent the regressors in the standard OFR algorithm \citep{Billings1989}, is the smallest number that achieves maximum accuracy on the validation dataset. In order to evaluate the effect of the sampling time $ \Delta T $ on the performance of the standard Readout, the training was performed for several sampling times ranging from $ 0.2 $ms to $ 30 $ms. The accuracies for all the methods, as a function of the sampling time, are depicted in Figure \ref{fig:example2_deltaT}. Each data point represents an average value over $ 10 $ different Liquids.

\begin{figure}[!ht]
	\hfill
	\begin{center}
		\includegraphics[width=5.5in,trim={0cm 0cm 1.5cm 0cm},clip]{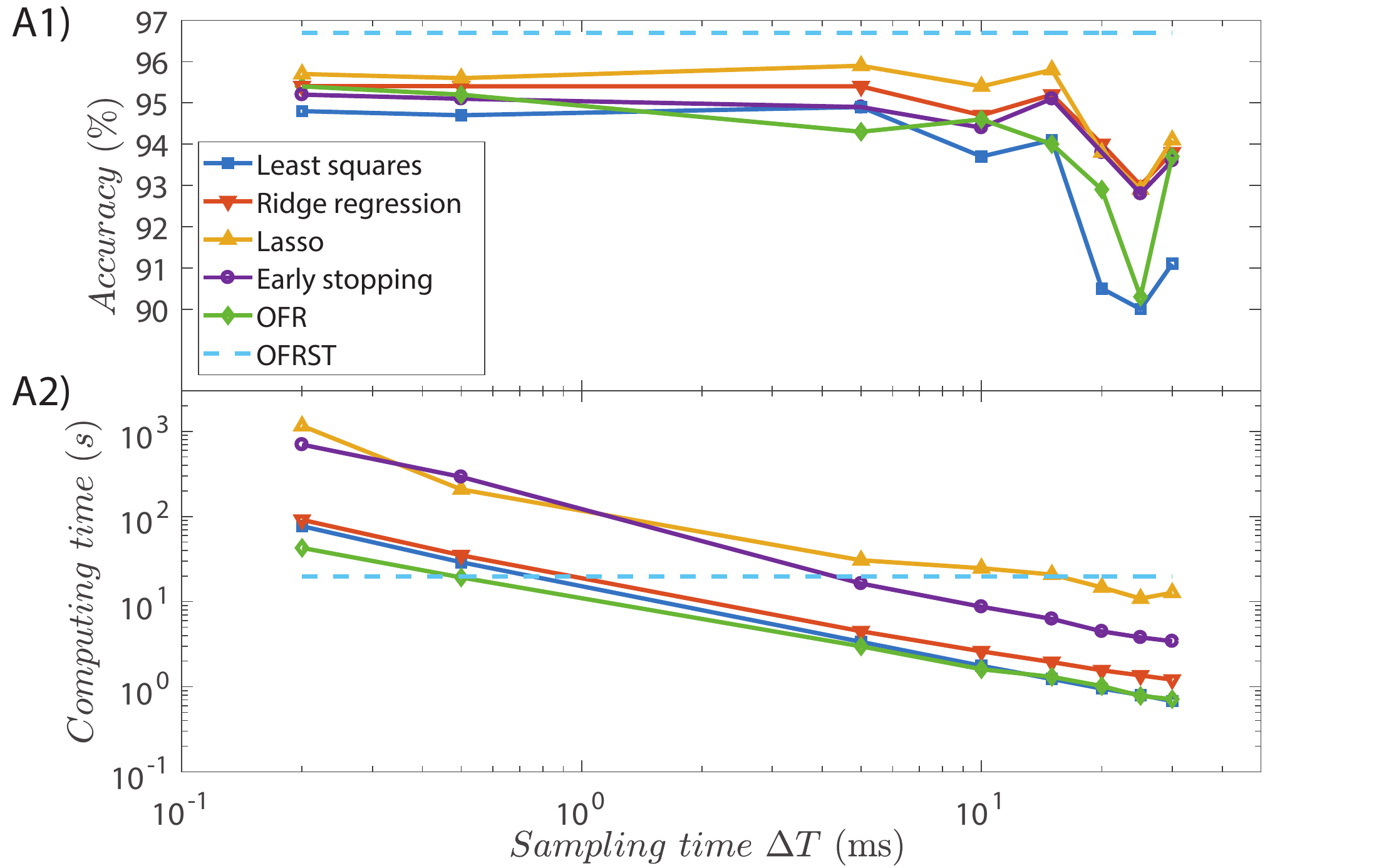}
	\end{center}
	\caption{Comparison between the proposed OFRST method and LS, RR, lasso, ES and OFR for different values of the sampling time $ \Delta T $: accuracies (A1) and computing times (A2).}
	\label{fig:example2_deltaT}
\end{figure}

The results show that the classification accuracy for the LS, RR, lasso, ES and classical OFR methods can be increased by decreasing the sampling time. However, the performance is still below the one achieved by the OFRST method, which selects presynaptic connections using the exact spike times generated by the Liquid neurons. The difference in accuracy between OFR and OFRST, which is expected to vanish when $ \Delta T \rightarrow 0 $, shows directly advantage in processing exact spike times.

Interestingly, even for $ \Delta T=0.2 $ ms, which is the sampling time used for simulating the LSM, OFRST still performs significantly better than the other methods. This is because all the training methods based on the standard Readout architecture are subject to an approximation error when estimating the weights, for any $ \Delta T >0 $.

\vspace{1cm}
\textbf{Example 3. Binary classification: selecting relevant presynaptic neurons.} 

This numerical example evaluates the performance of the OFRST in selecting the relevant presynaptic partners  using exact spike timing. 
The SNN used in this example has a reservoir consisting of two sub-networks that are disconnected from one another, each sub-network consisting of a different pool of $ 135 $ spiking neurons generated as in examples 1 and 2. Two templates were generated as Poisson spike trains with frequency of $ 20 $ Hz over interval $ [0, 0.5\text{s}]. $ The first pool $ R_1= \{r_1,\dots,r_{135}\} $ receives $ 200 $ inputs generated by jittering the two spike train templates, $ 100 $ for each class, of which $ 50 $ were used for training and $ 50 $ for validation. The jitter noise is drawn from the Gaussian distribution with zero mean and standard deviation $ 1 $ ms.
The second pool $ R_2=\{r_{136},\dots,r_{270}\} $ receives a number of $ 200 $ new jittered inputs generated from the same two templates but in a different order selected at random. 

The task is to classify the inputs to sub-network $ R_1 $ using the neuron outputs from the full reservoir. The OFRST algorithm is compared with the LS, RR, lasso, ES and the OFR algorithms, which use the standard filtered spike train outputs. 

In essence, when solving the binary classification problem, the algorithms should only select neurons from $ R_1 $ as pre-synaptic partners of the Readout unit. The training results, computed for $ 100 $ different Liquids and instances of jitter, are summarised in Table \ref{tb:neuron_selection}.

\begin{table}[ht]
	\renewcommand{\arraystretch}{0.55}
	\caption{Binary classification results for least squares, ridge regression, lasso, early stopping, standard OFR and OFRST using two unconnected sub-networks with $ 135 $ neurons each. The reported values represent means ($ \pm $ standard deviations) computed over $ 100 $ trials.}
	\begin{center}
		\begin{tabular}{|c|c|c||c|}
			\hline
			Training method & 
			\begin{tabular}{@{}c@{}}
				Total number\\ of Readout\\ connections  
			\end{tabular} & 
			\begin{tabular}{@{}c@{}}
				Percentage connections to\\ sub-network $ R_1 $ 
			\end{tabular}
			& Accuracy \\
			\hline
			Least squares & $ 67.6\ (\pm21.72) $  & $ 49.9\%\ (\pm1.6\%) $ & $ 95.7\%\ (\pm4.4\%) $\\
			Ridge regression & $ 67.6\ (\pm21.72) $  & $ 49.9\%\ (\pm1.6\%) $  & $ 97.1\%\ (\pm3.8\%) $\\
			Lasso & $ 49.6\ (\pm26.5) $  & $ 58.9\%\ (\pm14.19\%) $ & $ 97.6\%\ (\pm3.4\%) $\\
			Early stopping & $ 67.6\ (\pm21.72) $  & $ 49.9\%\ (\pm1.6\%) $  & $ 96.8\%\ (\pm4\%) $\\						
			OFR & $ 13.1\ (\pm10) $  & $ 86.7\%\ (\pm14.6\%) $ & $ 96.8\%\ (\pm4.2\%) $\\						
			OFRST & $ 9.45\ (\pm9.4) $  & $ 93.6\%\ (\pm10.7\%) $ & $ 97.8\%\ (\pm3.8\%) $\\									
			\hline
		\end{tabular}
	\end{center}
	\label{tb:neuron_selection}
\end{table}

The results show that OFRST achieves the highest accuracy among all methods using the least number of Readout presynaptic connections, and the highest percentage of connections to the correct sub-network $ R_1. $ Only OFR and OFRST achieve a percentage of connections to $ R_1 $ of over $ 90\% $, while all the other methods result in percentages just above chance.

\vspace{1cm}
\textbf{Example 4. Motion direction decoding using multi-electrode array recordings from the primary visual cortex.}

Here we use the proposed methodology to decode stimulus features using simultaneous multi-electrode array recordings of visually evoked activity from the primary visual cortex of three anesthetized macaque monkeys. The data were downloaded from the CRCNS online database \citep{Kohn2016}. Here we use the recordings from monkey number $ 1 $.

The stimuli were full-contrast drifting sinusoidal gratings at $ 12 $ orientations spaced equally ($ 0^{\circ}, 30^{\circ}, 60^{\circ}, \dots, 270^{\circ} $). Each stimulus was presented $ 200 $ times, for a duration of $ 1.3 $ s per trial \citep{Smith2008,Kelly2010}.
The spiking train responses of $ 106 $ neurons were simultaneously recorded using a Utah multi-electrode array and spike-sorted offline \citep{Smith2008,Kelly2010}. In this example we only use the first $ 200 $ ms from all recording trials, which is the time reported for visual categorisation tasks in primates \citep{FabreThorpe1998,Hung2005}. A recording trial for the $ 0^{\circ} $ drifting bar stimulus is depicted in Figure 5.

\begin{figure}[!ht] 	\label{fig:PVCdata}
	\begin{center}
		\includegraphics[width=5.5in,trim={2cm 0cm 0cm 0cm},clip]{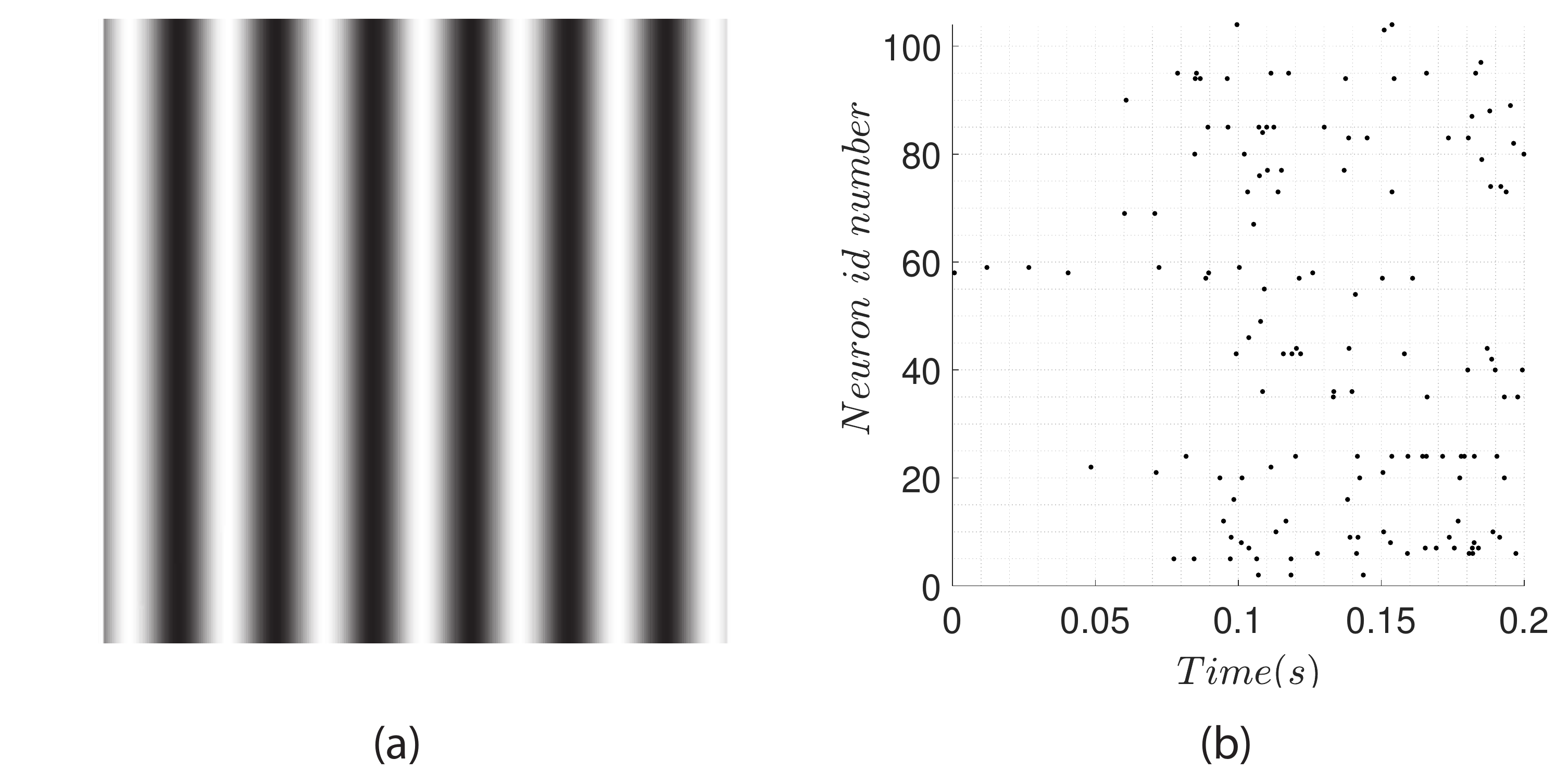}
	\end{center}
	\caption{The first sweep of experimental data used in Example 4: a) The first frame of a drifting bar stimulus oriented at $ 0^{\circ} $, b) Raster plot showing the response of $ 106 $ neurons, as a function of time.}
\end{figure}

The decoding task is to predict the stimulus orientation based on the recorded neural activity. The task is formulated as a multi-label classification problem. Each of the $ 12 $ directions was assigned a target label ($ 1-12 $) and a Readout. Each Readout processes the outputs of the $ 106 $ recorded neurons, which play the role of the Liquid spike train outputs.

The data ($ 2400 $ trials) was randomly divided into equal datasets for training and validation, such that each dataset comprises $ 100 $ trials with each of the $ 12 $ inputs. The $ 12 $ Readouts were trained using the "one-to-all" method, also known as "1-hot coding", where only one Readout generates an output "$ 1 $" at any given time.
Specifically, the target output for each Readout satisfies $ y^*(t)=1 $ when the input direction label matches the Readout label, and $ y^*(t)=-1 $ for any other direction. The overall prediction is given by the label of the Readout with maximum average value. The training data for each Readout consists of $ 100 $ trials from the target class and $ 100 $ trials evenly distributed among all other classes. The parameters of the $ 12 $ Readouts were tuned using the LS, RR, lasso, ES and the OFRST methods. Considering the large number of possible connections, here the OFRST algorithm for each Readout was stopped when the criterion $ ERR_p<\zeta $ was met, where $ \zeta=4\cdot 10^{-4} $ is a parameter determined using line search and $ ERR $ denotes the error reduction ratio (see Appendix 4). Essentially, this means that each Readout only connects to presynaptic neurons whose outputs contribute more than $ 0.04\% $ to the variance change in the target output. The regularization parameters for RR and lasso, and the number of iterations for ES were selected using line search to maximise the accuracy on the validation dataset.
The final accuracy, computed on the validation dataset, is defined as the percentage of correctly decoded input directions.

We compared the decoding performance with standard Readouts, trained with LS, RR, lasso and ES, to the performance with spike time based Readouts, trained with the OFRST algorithm described in subsection \ref{subs:OFRST}.
The results are summarised in tables \ref{tb:real_data_accuracy} and \ref{tb:real_data_connections}.

\begin{table}[ht]
	\renewcommand{\arraystretch}{0.55}
	\caption{Multi-label classification accuracies with the standard LS, RR, lasso, ES methods and the proposed OFRST method.}
	\begin{center}
		\begin{tabular}{|c|c|c|c|c|c|c|c|c|c|c|c|c|c|}
			\hline
			Training & \multicolumn{12}{|c|}{Readout accuracy (\%)} 
			& Final \\\cline{2-13}
			method & 1 &2&3&4&5&6&7&8&9&10&11&12&accuracy\\ 
			\hline
			LS & $ 86 $ & $ 84 $ & $ 83 $&$ 81 $&$ 82 $&$ 84 $&$ 85 $&$ 83 $&$ 84 $&$ 87 $&$ 80 $&$ 78 $&$ 58.17 \%$\\			
			\hline
			RR & $ 88 $ & $ 88 $ & $ 85 $&$ 87 $&$ 84 $&$ 80 $&$ 84 $&$ 88 $&$ 84 $&$ 85 $&$ 85 $&$ 82 $&$ 61.75 \%$\\			
			\hline		
			Lasso & $ 83 $ & $ 83 $ & $ 85 $&$ 87 $&$ 83 $&$ 82 $&$ 83 $&$ 90 $&$ 83 $&$ 86 $&$ 83 $&$ 81 $&$ 59.5 \% $\\
			\hline				
			ES & $ 89 $ & $ 85 $ & $ 87 $&$ 87 $&$ 82 $&$ 82 $&$ 88 $&$ 87 $&$ 88 $&$ 86 $&$ 82 $&$ 83 $&$ 61.33 \% $\\
			\hline			
			OFRST & $ 88 $ & $ 88 $ & $ 86 $&$ 86 $&$ 84 $&$ 86 $&$ 88 $&$ 86 $&$ 85 $&$ 89 $&$ 82 $&$ 84 $&$ 67.58 \%$\\
			\hline
		\end{tabular}
	\end{center}
	\label{tb:real_data_accuracy}
\end{table}

\begin{table}[ht]
	\renewcommand{\arraystretch}{0.55}
	\caption{Number of presynaptic connections selected for each Readout with the standard LS, RR, lasso, ES methods and the proposed OFRST method.}
	\begin{center}
		\begin{tabular}{|c|c|c|c|c|c|c|c|c|c|c|c|c|}
			\hline
			Training & \multicolumn{12}{|c|}{Number of Readout connections} 
			 \\\cline{2-13}
			method & 1 &2&3&4&5&6&7&8&9&10&11&12\\ 
			\hline
			LS & $ 106 $ & $ 106 $ & $ 106 $&$ 106 $&$ 106 $&$ 106 $&$ 106 $&$ 105 $&$ 105 $&$ 106 $&$ 106 $&$ 106 $\\
			\hline
			RR & $ 106 $ & $ 106 $ & $ 106 $&$ 106 $&$ 106 $&$ 106 $&$ 106 $&$ 106 $&$ 106 $&$ 106 $&$ 106 $&$ 106 $\\
			\hline
			Lasso & $ 69 $ & $ 71 $ & $ 62 $&$ 77 $&$ 69 $&$ 70 $&$ 68 $&$ 64 $&$ 75 $&$ 77 $&$ 72 $&$ 72 $\\
			\hline			
			ES & $ 106 $ & $ 106 $ & $ 106 $&$ 106 $&$ 106 $&$ 106 $&$ 106 $&$ 106 $&$ 106 $&$ 106 $&$ 106 $&$ 106 $\\
			\hline				
			OFRST & $ 49 $ & $ 61 $ & $ 63 $&$ 61 $&$ 66 $&$ 59 $&$ 58 $&$ 61 $&$ 60 $&$ 60 $&$ 61 $&$ 59 $\\
			\hline
		\end{tabular}
	\end{center}
	\label{tb:real_data_connections}
\end{table}

The results show that the proposed spike time based Readout, trained with the OFRST algorithm, performs significantly better than the standard Readout architecture trained with LS, RR, lasso, or ES, while using significantly fewer neuron connections. 

\vspace{1cm}
\textbf{Example 5. Speech recognition.}

In this example we use the proposed OFRST methodology to perform speech recognition. The data is a subset of the TI-46 corpus of isolated spoken digits, consisting of $ 500 $ utterances of digits "zero" to "nine" spoken by $ 5 $ different female speakers \footnote{\normalsize Downloaded from the Linguistic Data Consortium website: http://www.ldc.upenn.edu.} \citep{TI46_1,TI46_2}.

The decoding task is to predict the digit number using a LSM, formulated as a multi-label classification problem \citep{isolated2005}. The LSM in this example has $ 135 $ neurons, spatially organised as a lattice with dimensions $ 15 $x$ 3 $x$ 3 $.  As before, each digit was assigned a target label (1-10) and a Readout unit.

The data is preprocessed using the Lyon passive ear model, which is a model of the human inner ear, or cochlea. This model consists of three processing stages: a band-pass filter-bank, inspired by the human ear sensitivity to certain frequencies, half way rectification, and automatic gain control, which model the hair cells in the cochlea \citep{Lyon1982}. 
Subsequently, the continuous output of the Lyon passive ear model is converted into a spike train using an algorithm called Ben's spiker algorithm (BSA) \citep{Ben2003}. This preprocessing front-end, consisting of the Lyon passive ear model in series with BSA, has been used successfully to address this type of speech recognition problem using an LSM \citep{isolated2005,experimental2007,Yin2012}.

The data was divided in two sets: a training set of size $ 300 $ and a validation set of size $ 200 $, such that the recordings of each speaker are proportionally distributed between the two sets. As before, the $ 10 $ Readouts were trained using the "one-to-all" method. The training data for each Readout consists of $ 60 $ recordings from the corresponding target class and $ 60 $ recordings evenly distributed among all other classes. The final accuracy is defined as the percentage of correctly recognised digits in the validation dataset.

The $ 10 $ Readout units were trained using LS, RR, lasso, ES and the OFRST method. The stop criterion for the OFRST algorithm is $ ERR_{p+1}-ERR_{p}<\zeta $. The parameter $ \zeta $ and the regularization parameters for RR, lasso and ES were tuned for each Readout on the validation dataset using line search. 

The comparative performance of the spike time based Readouts trained with OFRST, and the standard Readouts trained with LS, RR, lasso and ES are summarised in tables \ref{tb:spoken_accuracy} and \ref{tb:spoken_connections}. 

\begin{table}[ht]
	\renewcommand{\arraystretch}{0.55}
	\caption{Multi-label classification accuracies for the LS, RR, lasso, ES methods and the proposed OFRST method, computed as mean ($ \pm $ standard deviation) for $ 10 $ different Liquid simulations.}
	\begin{center}
		\begin{tabular}{|c|c|c|c|c|c|c|}
	\hline
	Training & \multicolumn{6}{|c|}{Readout accuracy (\%)} \\\cline{2-7}
	method & 1 &2&3&4&5&6\\ 
	\hline
	LS & $ 88 (\pm6) $ & $ 96(\pm3) $ & $ 96(\pm2) $&$ 95(\pm4) $&$ 94(\pm4) $&$ 95(\pm4) $\\			
	\hline
	RR & $ 95 (\pm3) $ & $ 92(\pm3) $ & $ 91(\pm6) $&$ 93(\pm5) $&$ 95(\pm3) $&$ 93(\pm4) $\\		
	\hline		
	Lasso & $ 96 (\pm2) $ & $ 90(\pm2) $ & $ 87(\pm7) $&$ 92(\pm7) $&$ 95(\pm4) $&$ 91(\pm4) $\\
	\hline				
	ES & $ 96 (\pm3) $ & $ 92(\pm3) $ & $ 92(\pm5) $&$ 94(\pm6) $&$ 95(\pm3) $&$ 93(\pm5) $\\
	\hline			
	OFRST & $ 94 (\pm 4) $ & $ 95 (\pm 3) $ & $ 100 (\pm 1) $&$ 91 (\pm 4) $&$ 98 (\pm 2) $&$ 93 (\pm 7) $\\
	\hline
		\end{tabular}
	\end{center}
	\begin{center}
		\begin{tabular}{|c|c|c|c|c|c|c|c|}
	\hline
	Training & \multicolumn{4}{|c|}{Readout accuracy (\%)} 
	& Final  \\ \cline{2-5}
	method &7&8&9&10&accuracy\\ 
	\hline
	LS & $ 97(\pm3) $&$ 94 (\pm1) $&$ 92 (\pm5)$&$ 95 (\pm4) $&$ 73.4\% (\pm5.5\%)$\\			
	\hline
	RR &$ 100(\pm1) $&$ 92 (\pm4) $&$ 89 (\pm5)$&$ 91 (\pm2) $&$ 86.4\% (\pm 2.9\%)$\\			
	\hline		
	Lasso & $ 99(\pm1) $&$ 92 (\pm4) $& $ 83 (\pm8)$&$ 89 (\pm4) $&$ 85.9\% (\pm 2.2\%) $\\
	\hline				
	ES & $ 99(\pm1) $&$ 91 (\pm4) $&$ 89 (\pm4)$&$ 91 (\pm3) $&$ 86.6\% (\pm 2.4\%) $\\
	\hline			
	OFRST & $ 99 (\pm 1) $&$ 92 (\pm 5) $&$ 99 (\pm 2) $&$ 90 (\pm 3) $&$ 88\% (\pm 1.9\%)$\\
	\hline
		\end{tabular}
	\end{center}
	\label{tb:spoken_accuracy}
\end{table}

\begin{table}[ht]
	\renewcommand{\arraystretch}{0.55}
	\caption{The average number of presynaptic connections selected for each Readout using the LS, RR, lasso, ES methods and the proposed OFRST method, computed for $ 10 $ different Liquids.}
	\begin{center}
		\begin{tabular}{|c|c|c|c|c|c|c|c|c|c|c|c|c|}
			\hline
			Training & \multicolumn{10}{|c|}{Average number of Readout connections} 
			\\\cline{2-11}
			method & 1 &2&3&4&5&6&7&8&9&10\\ 
			\hline
			LS & $ 112 $ & $ 112 $ & $ 112 $&$ 112 $&$ 112 $&$ 112 $&$ 112 $&$ 112 $&$ 112 $&$ 112 $\\
			\hline
			RR & $ 112 $ & $ 112 $ & $ 112 $&$ 112 $&$ 112 $&$ 112 $&$ 112 $&$ 112 $&$ 112 $&$ 112 $\\
			\hline
			Lasso & $ 87 $ & $ 82 $ & $ 85 $&$ 86 $&$ 85 $&$ 86 $&$ 90 $&$ 88 $&$ 85 $&$ 85 $\\
			\hline			
			ES & $ 112 $ & $ 112 $ & $ 112 $&$ 112 $&$ 112 $&$ 112 $&$ 112 $&$ 112 $&$ 112 $&$ 112 $\\
			\hline				
			OFRST & $ 61 $ & $ 54 $ & $ 60 $&$ 59 $&$ 63 $&$ 67 $&$ 66 $&$ 68 $&$ 65 $&$ 58 $\\
			\hline
		\end{tabular}
	\end{center}
	\label{tb:spoken_connections}
\end{table}

The results show that the proposed spike based Readout architecture trained with OFRST leads to the highest final accuracy of correctly recognised spoken digits. Moreover, each spike based Readout trained with OFRST has significantly fewer connections to Liquid neurons compared to the corresponding standard Readout trained with LS, RR, lasso and ES. Relative to lasso, which results in the fewest presynaptic connections for the standard Readout, the proposed OFRST method leads to a total reduction of $ 28\% $ in number of connections to the Liquid.

\newpage
\thispagestyle{plain}
\mbox{}


\section{Conclusions}
\label{subs:conclusions}

This work proposed a spike based Readout architecture for LSMs and introduced a new training method that uses the exact spike timing information generated by SNN models, or recorded during experimental procedures. The new method implements an orthogonal forward regression algorithm for training the Readout parameters, which exploits a distance metric defined in a spike train space.

The new algorithm, called orthogonal forward regression with spike trains (OFRST), allows the selection of the connectivity between the Liquid and the Readout unit, i.e., the neurons in the Liquid that are particularly relevant for solving a given learning or decoding task.

One advantage is that computations are carried out directly on spike trains. The standard methods filter the spike trains and then perform uniform sampling in order to optimise the weights. It is demonstrated theoretically and shown through numerical simulations, with synthetic and experimental data, that the classification accuracy is improved by using exact spike times. 

Specifically, new theoretical results demonstrated that the proposed Readout trained with OFRST outperforms the standard Readout, which combines linearly the uniform samples from the neuron filtered outputs and is trained with ordinary least squares, ridge regression, lasso or early stopping. Numerical simulations with synthetic data confirmed the theoretical findings and also showed that the proposed algorithm leads to a much smaller number of Readout synapses. A numerical study showed that OFRST outperforms the standard methods on decoding the orientation of drifting gratings using the multi-electrode array recordings of the evoked activity in the primary visual cortex of the monkey. An additional example showed the advantage in using the OFRST method on a speech recognition task.

It is interesting to highlight the fact that typically around less than $ 20\% $ of the total possible connections between Liquid and Readout are required, and that fully connected Readouts achieve less accuracy  on classification tasks. This suggests that, besides decoding stimulus features from the evoked brain activity, the new training method could also be used to characterise the functional specificity of neurons in the brain.

\section*{Appendix 1. Proofs of theorems}
\begin{proof}[Proof of Theorem \ref{theo:mapping}]
	The mapping \eqref{eq:mapping} is well defined if the operator $ \cf:\Sp^{out}\rightarrow \mathbb{F}\mathbb{S}^{out}$ is well defined and invertible. 
	
	A function $ y\in \mathbb{F}\mathbb{S}^{out} $ satisfies 
	\begin{equation*}
	y(t)= \sum_{k=1}^{N} w_k\cf s_k^{out}(t)=\cf\left(\sum_{k=1}^{N}  w_k s_k^{out}\right)(t).
	\end{equation*}
	According to the definition of $ \Sp^{out} $ it follows that $ \sum_{k=1}^{N} w_k s_k^{out} \in \Sp^{out} $, and therefore $ \cf:\Sp^{out}\rightarrow \mathbb{F}\mathbb{S}^{out} $ is well defined. Moreover, $ \cf $ is invertible if it is a one-to-one and onto operator. Let $ s_1=\sum_{k=1}^{N} v_k s_k^{out} $ and $ s_2=\sum_{k=1}^{N} w_k s_k^{out}. $ Operator $ \cf $ is one-to-one if 
	\[\cf s_1=\cf s_2 \Rightarrow s_1=s_2.\]
	It follows that
	\[\cf s_1=\cf s_2 \Leftrightarrow \sum_{k=1}^{N} v_k\cf s_k^{out}(t)=\sum_{k=1}^{N} w_k\cf s_k^{out}(t)\Leftrightarrow \sum_{k=1}^{N} (v_k-w_k)\cf s_k^{out}(t)=0. \]
	
	The functions $ \{\cf s_k^{out}\}_{k=1}^N $ are linearly independent according to Remark \ref{rem:lin_ind}. It follows that $ w_k=v_k, \forall k=1,\dots,N $, and thus $ \cf $ is one-to-one. According to the definition of $ \mathbb{F}\mathbb{S}^{out} $ and due to the linearity of $ \cf $, it follows that $ \cf $ is also an onto operator, and thus it is invertible.
	
\end{proof}

\begin{proof}[Proof of Theorem \ref{theo:optimum}]
	\begin{align*}
	\| {y^*} - {\cR}_{\bs{w}} \cF \bs s^{out}\|^2_{L^2}
	&=\| {y^*} - {\cR}_{\bs{w}} \left[\cf_1 s_1^{out},\dots,\cf_N s_N^{out}\right]\|^2_{L^2}\\
	&=\| {y^*} - \sum_{n=1}^{N} w_n \cf s_n^{out}\|^2_{L^2}\\	
	&=\| {y^*} - \cf \sum_{n=1}^{N} w_n  s_n^{out}\|^2_{L^2}\\
	&=\| {y^*} - \cf \bar{\cR}_{\bs{w}}  \bs{s^{out}}\|^2_{L^2}\\
	&=\| {y^*} \|^2_{L^2}+ \|\cf \bar{\cR}_{\bs{w}}\bs{s^{out}}\|_{L^2}^2- 2\left\langle y^*,\cf \bar{\cR}_{\bs{w}}\bs{s^{out}}\right\rangle_{L^2} \\
	&=\| {y^*} \|^2_{L^2}+ \|\cf \bar{\cR}_{\bs{w}}\bs{s^{out}}\|_{L^2}^2- 2\left\langle \cP_{\mathbb{F}\Sp^{out}} (y^*),\cf \bar{\cR}_{\bs{w}}\bs{s^{out}}\right\rangle_{L^2} \\
	&=\| {y^*} \|^2_{L^2}+ \|\cf \bar{\cR}_{\bs{w}}\bs{s^{out}}\|_{L^2}^2- 2\left\langle \cf s^{y*},\cf \bar{\cR}_{\bs{w}}\bs{s^{out}}\right\rangle_{L^2} \\
	&=\| {y^*} \|^2_{L^2}+ \frac{1}{2}\|\bar{\cR}_{\bs{w}}\bs{s^{out}}\|_{\Sp}^2- \left\langle s^{y*}, \bar{\cR}_{\bs{w}}\bs{s^{out}}\right\rangle_{\Sp} \\	
	&= \frac{1}{2}\left[\|\bar{\cR}_{\bs{w}}\bs{s^{out}}\|_{\Sp}^2+\|s^{y*}\|_{\Sp}^2 - 2\left\langle s^{y*}, \bar{\cR}_{\bs{w}}\bs{s^{out}}\right\rangle_{\Sp}\right]\\ 
	&\hspace{5cm}-\frac{1}{2}\|s^{y*}\|_{\Sp}^2+\| {y^*} \|^2_{L^2} \\	
	&=\frac{1}{2}\| s^{y*} - \bar{\cR}_{\bs{w}}  \bs s^{out}\|_{\Sp}^2-\frac{1}{2}\|s^{y*}\|_{\Sp}^2+\| {y^*} \|^2_{L^2}.
	\end{align*}
\end{proof}

\section*{Appendix 2. The Standard Orthogonal Forward Regression Algorithm (OFR)}
\label{appendix2a}

Let $ ERR_j^{(p)} $ be the error reduction ratio corresponding to term $ x_j $ at iteration $ p $ defined as
\begin{equation*}
ERR_j^{(p)}=\frac{\left\langle x_j^{\perp(p)},y^* \right\rangle_{L^2}^2}{\|x_j^{\perp(p)}\|_{L^2}^2\cdot\|y^*\|_{L^2}^2}.
\end{equation*}

The algorithm for training the Readout weights and predicting the input class is given as follows.

\begin{itemize}
	\item Initialization
	\begin{itemize}
		\item $ x_j^{\perp(1)}=x_j ,j=1,\dots,N, $
		\item $ \ell_1=\underset{j\in\{1,\dots,N\}}{\textup{argmax}} ERR_j^{(1)},L^{(1)}=\{\ell_1\}, $
		
		\item $ ERR_1=ERR_{\ell_1}^{(1)}, $
		\item $ x^{\perp}_1=x_{\ell_1} ,w_1^{\perp}=\frac{\langle
			y^*,x_1^{\perp}\rangle_{L^2}}{\|x_1^{\perp}\|_{L^2}^2}, $
		\item 	$ w_1=w_1^{\perp}. $
		
	\end{itemize}
	
	\item For $ p=2,\dots,N $, compute:
	\begin{itemize}
		\item $ 
		x_j^{\perp (p)}=	x_j^{\perp(p-1)}-\frac{\langle
			x_j,x_{p-1}^{\perp}\rangle_{L^2}}{\|x_{p-1}^{\perp}\|_{L^2}^2} , j\in\{1,\dots,N\}\textbackslash L^{(p-1)}, $
		\item $ \ell_p=\underset{j\in\{1,\dots,N\}\textbackslash L^{(p-1)}}{\textup{argmax}} ERR_j^{(p)},L^{(p)}=L^{(p-1)}\cup\{\ell_p\}, $
		\item $ ERR_p=ERR_{\ell_p}^{(p)}, $
		\item $ x^{\perp}_p=x_{\ell_p} ,w_p^{\perp}=\frac{\langle
			y^*,x_p^{\perp}\rangle_{L^2}}{\|x_p^{\perp}\|_{L^2}^2}, $
		\item$ 	a_{i,p}=\frac{\langle x_i ,x_p^{\perp}\rangle_{L^2}}{\|x_p^{\perp}\|_{L^2}^2},i\in \{1,\dots,p-1\}, $
		\item $ 	\bs{A}^{(p)}=
		\left[ \begin{array}{cccc}
		1 & a_{1,2} & \dots & a_{1,p} \\
		0 & 1 & \dots & a_{2,p} \\
		\dots & \dots & \dots & \dots \\
		0 & 0 & \dots & a_{p-1,p} \\
		0 & 0 & \dots & 1 \end{array} \right], $
		\item $ \bs{w}^{\perp(p)}=[w_1^{\perp},\dots,w_p^{\perp}],	 $
		\item $ \bs{w}^{(p)}=\left[\bs{A}^{(p)}\right]^{-1}\bs{w}^{\perp(p)}, $\\		
		where $ \bs{w}^{(p)}=[w_1^{(p)},\dots,w_p^{(p)}] $ denote the Readout weights at iteration $ p $,
		\item $ \hat{y}^{(p)}=\sum_{k=1}^{p}w_k^{(p)}x_{\ell_p} , $\\
		where $ \hat{y}^{(p)} $ is the Readout output,
		\item $ 		Pred\left(\hat{y}^{(p)} \right)=\text{sign}\left[ \int_0^{T_{max}} \hat{y}^{(p)}(t) dt \right], $\\
		where $ Pred\left(\hat{y}^{(p)}\right) $ is the class prediction based on the Readout activity on time interval $ [0,T_{max}] $, and $ \text{sign}() $ denotes the sign function.
	\end{itemize}
	\item Select the smallest $ p $ that gives the minimum error for validation.
\end{itemize}

\section*{Appendix 3. The Orthogonal Forward Regression with Spike Trains Algorithm (OFRST)}
\label{appendix2}

Let $ ERR_j^{(p)} $ be the error reduction ratio corresponding to presynaptic neuron $ j $ at iteration $ p $ defined as
\begin{equation*}
	ERR_j^{(p)}=\frac{\left\langle s_j^{\perp(p)},s^{y*} \right\rangle_{\Sp}^2}{\|s_j^{\perp(p)}\|_{\Sp}^2\cdot\|s^{y*}\|_{\Sp}^2}.
\end{equation*}	

The target output spike train $ s^{y*} $ is unknown prior to training. However, for $ y^*(t)=\pm1 $, the inner product $ \left\langle s,s^{y*} \right\rangle_{\Sp}, \forall s \in \Sp, s=\{(a_k,t_k)\}_{k=1}^M, $ can be computed on given time interval $ [T_1, T_2] $, representing the total simulation time, as follows
\begin{align*}
	\left\langle s,s^{y*} \right\rangle_{\Sp} &=  2\left\langle \cf s,\cf s^{y*} \right\rangle_{L^2}\\ &= 2\left\langle \cf s,y^*\right\rangle_{L^2}\\ &= (\pm 1)\tau_s \sum_{k=1}^{M} a_k\left[e^{-\frac{\textup{max}\{T_1,t_k\}-t_k}{\tau_s}}-e^{-\frac{\textup{max}\{T_2,t_k\}-t_k}{\tau_s}}\right].
\end{align*}

The algorithm for training the Readout weights and predicting the input class is given as follows.

\begin{itemize}
	\item Initialization
	\begin{itemize}
		\item $ s_j^{\perp(1)}=s_j^{out},j=1,\dots,N, $
		\item $ \ell_1=\underset{j\in\{1,\dots,N\}}{\textup{argmax}} ERR_j^{(1)},L^{(1)}=\{\ell_1\}, $

		\item $ ERR_1=ERR_{\ell_1}^{(1)}, $
		\item $ s^{\perp}_1=s_{\ell_1}^{out},w_1^{\perp}=\frac{\langle
	s^{y*},s_1^{\perp}\rangle_{\Sp}}{\|s_1^{\perp}\|_{\Sp}^2}, $
		\item 	$ w_1=w_1^{\perp}. $

	\end{itemize}

	\item For $ p=2,\dots,N $, compute:
	\begin{itemize}
		\item $ s_j^{\perp (p)}=	s_j^{\perp(p-1)}-\frac{\langle s_j^{out},s_{p-1}^{\perp}\rangle_{\Sp}}{\|s_{p-1}^{\perp}\|_{\Sp}^2},j\in\{1,\dots,N\}\textbackslash L^{(p-1)}, $
		\item $ \ell_p=\underset{j\in\{1,\dots,N\}\textbackslash L^{(p-1)}}{\textup{argmax}} ERR_j^{(p)},L^{(p)}=L^{(p-1)}\cup\{\ell_p\}, $
		\item $ ERR_p=ERR_{\ell_p}^{(p)}, $
		\item $ s^{\perp}_p=s_{\ell_p}^{out},w_p^{\perp}=\frac{\langle
	s^{y*},s_p^{\perp}\rangle_{\Sp}}{\|s_p^{\perp}\|_{\Sp}^2}, $
		\item$ 	a_{i,p}=\frac{\langle s_i^{out},s_p^{\perp}\rangle_{\Sp}}{\|s_p^{\perp}\|_{\Sp}^2},i\in \{1,\dots,p-1\}, $
		\item $ 	\bs{A}^{(p)}=
	\left[ \begin{array}{cccc}
	1 & a_{1,2} & \dots & a_{1,p} \\
	0 & 1 & \dots & a_{2,p} \\
	\dots & \dots & \dots & \dots \\
	0 & 0 & \dots & a_{p-1,p} \\
	0 & 0 & \dots & 1 \end{array} \right], $
		\item $ \bs{w}^{\perp(p)}=[w_1^{\perp},\dots,w_p^{\perp}],	 $
		\item $ \bs{w}^{(p)}=\left[\bs{A}^{(p)}\right]^{-1}\bs{w}^{\perp(p)}, $\\		
where $ \bs{w}^{(p)}=[w_1^{(p)},\dots,w_p^{(p)}] $ denote the Readout weights at iteration $ p $,
		\item $ \hat{s}^{(p)}=\sum_{k=1}^{p}w_k^{(p)}s_{\ell_p}^{out}, $\\
where $ \hat{s}^{(p)} $ is the Readout output,
		\item $ 		Pred\left(\hat{s}^{(p)} \right)=\text{sign}\left[T_{max}-2\tau_s\sum_{k=1}^{M_p}a_k^{(p)}\left(e^{-\frac{T_{max}-t_k^{(p)}}{\tau_s}}-1\right)\right], $\\
where $ Pred\left(\hat{s}^{(p)}\right) $ is the class prediction based on the Readout activity on time interval $ [0,T_{max}] $, $ \text{sign}() $ denotes the sign function, and $ \hat{s}^{(p)}=\left\lbrace\left(a_k^{(p)},t_k^{(p)}\right)\right\rbrace_{k=1}^{M_p}. $
\end{itemize}
		\item Select the smallest $ p $ that gives the minimum error for validation.
\end{itemize}

\section*{Acknowledgments}

DF and DC gratefully acknowledge that this work was supported by BBSRC under grant BB/M025527/1. We also gratefully acknowledge reviewers' comments, which helped improve the quality of the manuscript.

\end{document}